\documentclass[11pt,letterpaper]{article}

\usepackage{amsmath,amssymb,mathtools} 
\usepackage{algorithm}
\usepackage[noend]{algpseudocode}      
\algrenewcommand\algorithmiccomment[1]{\hfill\(\triangleright\) #1}

\usepackage[margin=1in]{geometry}

\usepackage{fullpage}

\usepackage[subtle,mathspacing=normal]{savetrees}
\usepackage{url}
\usepackage{amsmath}
\usepackage{algorithm}
\usepackage[noend]{algpseudocode}
\usepackage{amssymb}
\usepackage{amsthm}
\usepackage{color}
\usepackage{array}
\usepackage{xy}
\usepackage{setspace}
\usepackage{varwidth}
\usepackage{multicol}
\usepackage{algorithm}
\usepackage{hyperref}
\usepackage[capitalize]{cleveref}
\usepackage{mathtools}
\usepackage{cite}

\usepackage{comment}

\newtheorem{theorem}{Theorem}[section]

\newtheorem{lemma}[theorem]{Lemma}

\newtheorem{claim}[theorem]{Claim}

\theoremstyle{remark}

\crefname{theorem}{Theorem}{Theorems}
\crefname{lemma}{Lemma}{Lemmas}
\crefname{claim}{Claim}{Claims}
\crefname{observation}{Observation}{Observations}

\usepackage{enumitem}
\setlist[enumerate]{label=\arabic*.}

\usepackage{thm-restate}

\begin{document}

\title{Improved Online Sorting}
 \author{Jubayer Nirjhor\thanks{\texttt{nirjhor@umich.edu}}\\University of Michigan \and Nicole Wein\thanks{\texttt{nswein@umich.edu}}\\University of Michigan}

\date{}
\maketitle

\pagenumbering{gobble}
\begin{abstract}
    We study the online sorting problem, where $n$ real numbers arrive in an online fashion, and the algorithm must immediately place each number into an array of size $(1+\varepsilon) n$ before seeing the next number. After all $n$ numbers are placed into the array, the cost is defined as the sum over the absolute differences of all $n-1$ pairs of adjacent numbers in the array, ignoring empty array cells. Aamand, Abrahamsen, Beretta, and Kleist introduced the problem and obtained a deterministic algorithm with cost $2^{O\left(\sqrt{\log n \cdot\log\log n +\log \varepsilon^{-1}}\right)}$, and a lower bound of $\Omega(\log n / \log\log n)$ for deterministic algorithms. We obtain a deterministic algorithm with quasi-polylogarithmic cost $\left(\varepsilon^{-1}\log n\right)^{O\left(\log \log n\right)}$.

    Concurrent and independent work by Azar, Panigrahi, and Vardi achieves polylogarithmic cost $O(\varepsilon^{-1}\log^2 n)$.
\end{abstract}

\clearpage

\pagenumbering{arabic}

\section{Introduction}
In the \emph{online sorting problem}, $n$ real numbers arrive in an online fashion, and the algorithm must immediately place each number into an array of size $(1+\varepsilon) n$ before seeing the next number (so that each array cell contains at most one number). After all $n$ numbers are placed into the array, the \emph{cost} is defined as the sum over the absolute difference of all $n-1$ pairs of adjacent numbers in the array, ignoring empty array cells. Formally, let $A$ be the array of size $(1+\varepsilon)n$ where we place input numbers, indexed starting from $1$, and let $1\leq i_1\leq i_2\leq \cdots\leq i_n\leq (1+\varepsilon)n$ be the subsequence of $n$ array indices that are nonempty at the end (an input number was placed there). Then we define the cost as $\sum_{k=1}^{n-1}\left|A[i_k]-A[i_{k+1}]\right|$.

It is evident that the minimum achievable cost occurs when the numbers are perfectly sorted. This natural problem was introduced by Aamand, Abrahamsen, Beretta, and Kleist~\cite{aabk} in the context of online geometric packing problems. 

It is convenient to normalize the problem to assume that the $n$ numbers come from the range $[0,1]$ and contain the numbers 0 and 1, because then the cost of the optimal offline algorithm is precisely 1, so the competitive ratio of an online algorithm is equal to its cost. We use this convention unless otherwise specified.

For the case of $\varepsilon=0$ (where the array is of size $n$, and thus completely full at the end), an asymptotically tight competitive ratio of $\Theta(\sqrt{n})$ is known, where the upper bound is a deterministic algorithm~\cite{aabk}, and the lower bound holds even for randomized algorithms~\cite{MR4802290,aabk}. 

For the general setting of $\varepsilon>0$, Aamand, Abrahamsen, Beretta, and Kleist~\cite{aabk} obtained a deterministic algorithm with cost of $2^{O\left(\sqrt{\log n\cdot \log\log n +\log \varepsilon^{-1}}\right)}$. They also obtained, for any constant $\varepsilon$, a lower bound of $\Omega(\log n / \log\log n)$ for deterministic algorithms (there is no known lower bound for randomized algorithms). This leaves a large gap between upper and lower bounds, and our work shrinks the gap.

\subsection{Our Result} We obtain an algorithm with \emph{quasi-polylogarithmic} competitive ratio:

\begin{restatable}{theorem} {mainThm}\label{mainThm}
   Given a positive integer $n$ and a real number $\varepsilon \in (0,3]$, there is a deterministic algorithm for online sorting with $n$ numbers, an array of size $(1+\varepsilon)n$, and competitive ratio $\left(\varepsilon^{-1}\log n\right)^{O(\log \log n)}$. 
\end{restatable}

\noindent\textbf{Remark.} Ignoring the $\varepsilon$ dependence for simplicity, our competitive ratio can also be written as $2^{O\left((\log\log n)^2\right)}$, where the exponent $O((\log \log n)^2)$ is an exponential improvement over the exponent $O(\sqrt{\log n\cdot \log\log n})$ of prior work. Further improvement of the exponent from $O((\log \log n)^2)$ to $O(\log \log n)$ would make the final cost polylogarithmic. 

\paragraph{Concurrent and Independent Work.}

In independent work, Azar, Panigrahi, and Vardi achieve a polylogarithmic bound \cite{online} using techniques orthogonal to ours. Specifically, for any $\varepsilon \geq \Omega(\log n/n)$, they give a deterministic algorithm with competitive ratio $O(\varepsilon^{-1}\log^2 n)$. Additionally, they provide a trade-off for sparser arrays: For any $\gamma \in [O(1), O(\log^2 n)]$, for an array with $\gamma n$ space they give a deterministic algorithm with competitive ratio $O(\gamma^{-1}\log^2 n)$.

\subsection{Related Work} Several variants of online sorting have been studied, including the stochastic~\cite{MR4802290,stoch,stoch2} and high-dimensional~\cite{MR4802290} versions. Online sorting is also related to geometric packing problems: lower bounds for online sorting have been used to provide lower bounds for online packing of convex polygons~\cite{aabk}. 

Online sorting is also related to other online array maintenance problems such as \emph{list labeling}~\cite{itai1981sparse,MR3109079,MR4537273,MR4849329}, \emph{matching on the line}~\cite{MR2995327,MR3824311,MR3984884,balkanski2023power}, and \emph{B-tree leaf splitting}~\cite{MR660704}. 

\subsection{Our Techniques}
We use the techniques of~\cite{aabk} as a starting point. First we will provide a high-level outline of their algorithm, and then describe the additional ideas behind our algorithm.

\paragraph{Algorithm of~\cite{aabk}.} Divide the interval $[0,1]$ into contiguous disjoint subintervals, and divide the array into contiguous disjoint subarrays called \emph{boxes} (for this outline, we will not specify how many subintervals and boxes). For any real number $x\in[0,1]$, let $I(x)$ be the subinterval containing $x$. 

Note that after all $n$ numbers are inserted, a $1/(1+\varepsilon)$ fraction of the array cells will be occupied. When the first number $x_0$ arrives, place $x_0$ into the first box and \emph{label} the box with $I(x_0)$. In general, when a number $x$ arrives, if there is a box labeled with $I(x)$ that is below the \emph{fullness threshold} of $1/(1+\varepsilon)$ (i.e.~a $<1/(1+\varepsilon)$ fraction of array cells in the box are occupied), place $x$ into that box. 
Otherwise, place $x$ into the first completely empty box and label that box $I(x)$. The parameters (such as the size of each box, the fullness threshold of each box, and the number of subintervals) are set so that there always exists a box to place $x$.

After placing $x$ into a box, the algorithm still needs to determine which array cell within the box to place $x$. This is done recursively. 
In order to ultimately achieve fullness $1/(1+\varepsilon)$, the recursive instance within each box must run the algorithm with fullness threshold $>1/(1+\varepsilon)$. 
This is because for any box, some of its sub-boxes may be almost completely empty (e.g.~with only one occupied array cell if no numbers from that particular subinterval ever arrive again in the input sequence), and to compensate, the other boxes must be fuller, hence a higher fullness threshold. Eventually, at the bottom recursive level, the fullness threshold becomes nearly 1, and the known $\Theta(\sqrt{n})$-cost algorithm for the $\varepsilon=0$ case is applied. Ultimately, the value of $\varepsilon$ that determines the fullness threshold on each recursive level is set to increase linearly as you go up the recursive levels.

\paragraph{Our Algorithm.} Our improved algorithm has two key ideas:

\textbf{Idea 1:} The first key idea is to use \emph{another layer of recursion} to decide in which box to place each number. In the algorithm of~\cite{aabk}, when they label a box, they simply choose the first available box. Instead, our algorithm treats the ``array of boxes'' as its own recursive instance of the online sorting problem. In other words, when a number $x$ arrives, the first step is the same as the algorithm of~\cite{aabk}: if there is a box labeled $I(x)$ that is below the fullness threshold, place $x$ into that box choosing the specific array cell recursively. If there is no such box, our algorithm deviates from theirs: we recursively run the algorithm on the array of boxes to determine which empty box to place $x$, and then choose the specific array cell within the box with a second recursive call. 

\textbf{Idea 2:} The second key idea is to increase $\varepsilon$ \emph{exponentially} instead of linearly as we proceed up the recursive levels. 
This is necessary because, for technical reasons, linear scaling of $\varepsilon$ clashes with the implementation of Idea 1. 
The reason exponential scaling works better than linear is roughly as follows. In order to treat the array of boxes as a recursive instance of online sorting with $\varepsilon>0$, some of the boxes need to be completely empty. 
Intuitively, varying $\varepsilon$ exponentially allows us to designate more empty boxes in the top layers of recursion, which makes the recursive algorithm on the array of boxes have lower cost.

\section{Proof of~\cref{mainThm}}

Recall~\cref{mainThm}:


\mainThm*

The bulk of the proof will be in proving the following lemma:

\begin{lemma}\label{mainLemma}
Consider positive integers $n$, $k$, and real numbers $\delta \in (0, \frac{1}{2})$, $[\alpha, \alpha + \beta) \subseteq [0, 1]$ such that $2^{k+1}\delta \leq 3$ is satisfied, and $\delta$, $k$ may depend on $n$. Then, there exists an algorithm that solves the online sorting problem on an array of size $(1+2^{k+1}\delta)n$, over any sequence of $n$ real numbers from $[\alpha, \alpha + \beta)$, achieving a cost of $\beta n^{1/\omega_k}\delta^{-O(k)}$. Here, the integer sequence $\{\omega_i\}_{i\in \mathbb Z}$ is defined recursively by $\omega_i=\omega_{i-1}+\omega_{i-4}$ for $i\geq 2$, and $\omega_i=2$ for $-\infty < i\leq 1$. 
\end{lemma}

Before proving \cref{mainLemma}, we first show that \cref{mainLemma} indeed implies \cref{mainThm}.

\begin{proof}[Proof of \cref{mainThm} assuming \cref{mainLemma}]
First, we look at the recurrence $\omega_i=\omega_{i-1}+\omega_{i-4}$ for $i\geq 2$ starting with $\omega_i=2$ for $-\infty < i\leq 1$. It's a linear recurrence with characteristic polynomial $x^4=x^3+1$. We note that, among the four complex roots of this polynomial, only the root $x\approx 1.38$ satisfies $|x|\geq 1$. So for every other root $x$, the value $x^n\to 0$ as $n\to \infty$. From this, we can deduce that $\omega_n\propto 1.38^n$ as $n\to \infty$. 

We apply \cref{mainLemma} choosing the parameters $\alpha\coloneqq 0$, $\beta\coloneqq 1$, $k\coloneqq \left\lfloor\log_{1.38} \log n\right\rfloor$, and $\delta\coloneqq \frac{\varepsilon}{2^{k+1}}$, since they satisfy $2^{k+1}\delta =\varepsilon\leq 3$ (by the assumption on the value of $\varepsilon$). It uses $(1+2^{k+1}\delta)n=(1+\varepsilon)n$ memory cells and yields a cost of \[n^{1/\omega_k}\delta^{-O(k)}=2^{O\left((\log n)/\omega_k + k\log(2^{k+1}/\varepsilon)\right)}=2^{O\left((\log n)/1.38^k+k^2+k\log\varepsilon^{-1}\right)}=(\log n)^{O(\log \log n+\log \varepsilon^{-1})}\] obtained by plugging in our chosen parameters. Using the identity $a^{\log b}=b^{\log a}$, we can infer that $(\log n)^{O\left(\log \varepsilon^{-1}\right)}=\left(\varepsilon^{-1}\right)^{O(\log \log n)}$. So the cost above is $\left(\varepsilon^{-1}\log n\right)^{O(\log \log n)}$. 
\end{proof}

With this, it remains to prove \cref{mainLemma}.

\subsection{Proof of \cref{mainLemma}}
We divide the proof into three sections. First, we present the algorithm itself. Afterwards, we prove that the proposed algorithm correctly solves the online sorting problem. Lastly, we prove that the algorithm indeed achieves a cost of $\beta n^{1/\omega_k} \delta^{-O(k)}$.

\subsubsection{The Algorithm}

We will describe the algorithm both in words and in pseudocode (Algorithm~\ref{sorter}).

\textbf{Algorithm setup and definitions.} Define $\mathcal{A}$ as the algorithm described in \cite{aabk} for online sorting with $n$ numbers on an array of size $n$, that achieves the asymptotically tight competitive ratio of $\Theta(\sqrt n)$. First off, we note that the algorithm $\mathcal{A}$ works even if the array size is bigger than $n$, as we could simply ignore the extra space. For a positive integer $k$, we will recursively define $\textsc{Sorter}_k$ as our algorithm that solves the online sorting problem with $n$ numbers on an array of size $(1+2^{k+1}\delta)n$. As $k$ gets larger, so does the amount of empty space in the array. Intuitively, this provides us with more flexibility, allowing for a smaller cost. We further define $\textsc{Sorter}_k\coloneqq \mathcal{A}$ for all integers $k\leq 1$.

We design the algorithm $\textsc{Sorter}_k$ for $k\geq 2$, by recursively using $\textsc{Sorter}_j$ for $j<k$ as subroutines. We will divide the array into $\ell$ disjoint contiguous subarrays that we call \emph{boxes}, each of \emph{width} $w$ (i.e.~each box has $w$ array cells). Each box can receive at most $n'<w$ elements of the input sequence. When we receive a particular element of the input sequence, we will invoke an instance of our algorithm to recursively determine the box where the element will go. We will also invoke an instance of our algorithm to recursively determine the exact cell in that box where the element should be placed. We define $n':=\left\lfloor\frac{2^{k-1}\delta}{1+2^k\delta}n^{\omega_{k-1}/\omega_k}\right\rfloor$ (the number of elements that go into a certain box), and $w:=\lfloor (1+2^k\delta)n'\rfloor$ (the width of each box). Then the total number of boxes that completely fit into the array is $\ell\coloneqq\left\lfloor\frac{(1+2^{k+1}\delta)n}{w}\right\rfloor$ (if the length of the array $(1+2^{k+1}\delta )n$ is not divisible by $w$, then we ignore the leftover space). For an integer $1\leq i\leq \ell$, we define the box $B_i$ as the subarray spanning the array indices from $(i-1)w$ to $iw-1$. We say that a box $B_i$ is \textit{full} if it contains exactly $n'$ elements of the input sequence. Intuitively, we call it full because we've reached the maximum number of elements we allow inserting into this particular subarray, while keeping enough empty space to recursively invoke our algorithm. 

We will also divide the interval of real numbers $[\alpha,\alpha+\beta)$ into $b\coloneqq \lfloor n^{\omega_{k-4}/\omega_k}\rfloor$ disjoint contiguous subintervals. Each box will be designated a particular subinterval of real numbers, so that only elements from that subinterval will be placed into the box. To determine which subinterval should be assigned to a particular box, we invoke our algorithm recursively. Let the algorithm $\textsc{BoxSorter}$ be a fixed instance of $\textsc{Sorter}_{k-4}$ on an array of size $\ell$. Intuitively, the $i$-th array cell here corresponds to the $i$-th box. We specifically use $\textsc{Sorter}_{k-4}$ here because, as we'll see in the analysis, this instance is enough to guarantee that all elements of the input sequence will find an array slot following our algorithm (allowing more empty space through, say, $\textsc{Sorter}_{k-3}$ or higher, does not work in our analysis). The input to $\textsc{BoxSorter}$ will be a subsequence of the original input sequence -- whenever we encounter an input from a subinterval that isn't assigned to a not-full box, we will append that element to the input sequence of $\textsc{BoxSorter}$. Then we'll insert the element into the box corresponding to the array cell where $\textsc{BoxSorter}$ places this element. In the analysis, we will show that an instance of $\textsc{Sorter}_{k-4}$ does indeed work as our $\textsc{BoxSorter}$ algorithm. 

We define the \emph{pointer function} $p:[b]\longrightarrow\mathbb{Z}$ and initialize $p(i)=-1$ for each $i\in [b]$. This function maps each of the $b$ subintervals to one of the $\ell$ boxes (with $-1$ indicating that no box has been assigned to it yet). 

\textbf{Final description of algorithm.} Given a real number $x\in [\alpha,\alpha+\beta)$, let $i\in [b]$ be the subinterval that $x$ falls into, that is, $i$ is the unique integer satisfying $(i-1)\beta \leq (x-\alpha)b < i\beta$. If $p(i)\neq -1$ and $B_{p(i)}$ is not full (meaning it has received strictly fewer than $n'$ elements), we place $x$ into $B_{p(i)}$. Otherwise, we send $x$ to $\textsc{BoxSorter}$ as an input and assign $p(i)$ to the index in $[\ell]$ where $x$ is placed by $\textsc{BoxSorter}$, then finally we place $x$ into $B_{p(i)}$. To place $x$ into an array cell inside the box $B_{p(i)}$, we run $\textsc{Sorter}_{k-1}$ on the array defined by the box and send $x$ as an input.

To summarize the recursive structure of the algorithm, $\textsc{Sorter}_k$ spawns one instance (that we call \textsc{BoxSorter}) of $\textsc{Sorter}_{k-4}$ on an array of size $\ell$, whose input values come from an interval of length $\beta$. It also spawns $\ell$ instances of $\textsc{Sorter}_{k-1}$, each on an array of size $w$ and receiving at most $n'$ elements, whose input values come from an interval of length $\frac{\beta}{b}$ (since we divide the original interval $[\alpha,\alpha+\beta)$ of length $\beta$ into $b$ subintervals of equal size). The input sequences for all these instances are subsequences of the given input sequence.

\begin{algorithm}[h!]
\caption{\textsc{Sorter}$_k$}\label{sorter}
\begin{algorithmic}[1]
\Require $k\in\mathbb{Z}$; array of size $(1+2^{k+1}\delta)n$; input interval $[\alpha,\alpha+\beta)$
\If{$k \le 1$}
  \State \Return $\mathcal{A}$ \Comment baseline online sorter from \cite{aabk}
\EndIf
\State $n' \gets \left\lfloor \dfrac{2^{k-1}\delta}{1+2^k\delta}\,n^{\omega_{k-1}/\omega_k} \right\rfloor$ \Comment{fullness of each box}
\State $w \gets \lfloor (1+2^k\delta)\,n' \rfloor$ \Comment{size of each box}
\State $\ell \gets \left\lfloor \dfrac{(1+2^{k+1}\delta)n}{w} \right\rfloor$\Comment{total number of boxes}
\State $b \gets \lfloor n^{\omega_{k-4}/\omega_k} \rfloor$ \Comment{number of subintervals}
\State $p[1..b] \gets -1$ \Comment pointer: subinterval $\to$ box
\State $count[1..\ell] \gets 0$ \Comment occupancy of each box
\State \textbf{instantiate} $\textsc{BoxSorter} \gets \textsc{Sorter}_{k-4}$ on array of size $\ell$
\For{$j \gets 1$ \textbf{to} $\ell$}
  \State \textbf{instantiate} $\textsc{InBoxSorter}[j] \gets \textsc{Sorter}_{k-1}$ on array of size $w$
\EndFor

\ForAll{arriving element $x \in [\alpha,\alpha+\beta)$}
  \State $i \gets 1 + \left\lfloor \dfrac{x-\alpha}{\beta} \cdot b\right\rfloor $ \Comment subinterval index in $[b]$
  \If{$p[i] \neq -1$ \textbf{and} $count[p[i]] < n'$}
    \State $j \gets p[i]$ \Comment reuse assigned not-full box
  \Else
    \State $j \gets \Call{\textsc{BoxSorter}}{x}$ \Comment pick a box index $j \in [\ell]$
    \State $p[i] \gets j$
  \EndIf
  \State \Call{$\textsc{InBoxSorter}[j]$}{x} \Comment place $x$ inside $B_j$
  \State $count[j] \gets count[j] + 1$
\EndFor
\end{algorithmic}
\end{algorithm}








    
        
        

\subsubsection{Proof of Correctness}
In this section, we prove, by induction on $k$, that our algorithm $\textsc{Sorter}_k$ successfully places the given input sequence of $n$ numbers into the array of size $(1+2^{k+1}\delta)n$. The base case of $k=1$ is immediate, since we define $\textsc{Sorter}_1\coloneqq \mathcal{A}$ as the algorithm described in \cite{aabk} to solve online sorting with $n$ numbers on an array of size $n$ (our array size is bigger than $n$). 

By the inductive hypothesis, for a positive integer $k\geq 2$, $\textsc{Sorter}_{k-1}$ can place $n'$ numbers into an array of size $w\coloneqq \lfloor (1+2^k \delta)n'\rfloor$. We observe that the only possible ways for us to not successfully insert an element are: either if $\textsc{BoxSorter}$ has reached the capacity (defined by a fraction of its size) and does not have a cell (corresponding to a box) to assign, or if we cannot find a cell inside the box assigned to the subinterval of the input element. But finding the cell within a particular box is handled recursively, and from the inductive hypothesis, we do know that $\textsc{Sorter}_{k-1}$ can place $n'$ elements in each box. So we only need to show that the input sequence to $\textsc{BoxSorter}$ is small enough for it to never run out of cells (corresponding to boxes) to assign. Recall that we use an instance of $\textsc{Sorter}_{k-4}$ as our $\textsc{BoxSorter}$ algorithm, so we will show that the number of elements sent to $\textsc{BoxSorter}$ respects the fraction of occupied cells required by $\textsc{Sorter}_{k-4}$. The proof follows immediately from the claim below, as we will see. 

\begin{claim}\label{boxsort}
Let $\mathcal{S}$ be the sequence of numbers we send to $\textsc{BoxSorter}$ as input. Then $|\mathcal S|\leq \frac{\ell}{1+2^{k-3}\delta}$.
\end{claim}

\begin{proof}[Proof of \cref{boxsort}]
Consider an element $x$ from the original input sequence that we include in $\mathcal S$. Depending on the type of box where $x$ is placed, we consider two cases.
\begin{enumerate}
    \item \emph{The box is full at the end.} Since a full box receives exactly $n'$ elements, the number of full boxes is at most $\frac{n}{n'}$. So the number of such elements in $\mathcal S$ is also at most $\frac{n}{n'}$ (since $\mathcal{S}$ only contains at most one element from each box).
    \item \emph{The box is not full at the end.} Since each non-full box is assigned a different subinterval, the number of such boxes at the end is at most $b$ (the number of subintervals), and so is the number of such elements in $\mathcal S$. 
\end{enumerate}
We conclude that $|\mathcal S|\leq \frac{n}{n'}+b$. Next we note that 
\begin{align*}
    bw\leq n^{\omega_{k-4}/\omega_k}\cdot (1+2^k\delta)n'\leq n^{\omega_{k-4}/\omega_k}\cdot (1+2^k\delta) \frac{2^{k-1}\delta}{1+2^k\delta} n^{\omega_{k-1}/\omega_k}=2^{k-1}\delta n
\end{align*}
using $\omega_k=\omega_{k-1}+\omega_{k-4}$. This bound helps us compute:
\begin{align*}
    \frac{|\mathcal S|(1+2^{k-3}\delta)}{\ell} &\leq \frac{(\frac{n}{n'}+b)(1+2^{k-3}\delta)}{\frac{(1+2^{k+1}\delta)n}{w}}\\ &= \frac{(\frac{n}{n'}w+bw)(1+2^{k-3}\delta)}{(1+2^{k+1}\delta)n} \\ 
    &\leq \frac{(\frac{n}{n'}(1+2^k\delta)n'+2^{k-1}\delta n)(1+2^{k-3}\delta)}{(1+2^{k+1}\delta)n}\\&=\frac{(1+3\cdot 2^{k-1}\delta )(1+2^{k-3}\delta)}{(1+2^{k+1}\delta)}.
\end{align*}
When $2^{k+1}\delta \leq 3$, this final quantity is $\leq 1$, which finishes the proof.

As a final remark, we note that we require $\varepsilon=2^{k+1}\delta \leq 3$ for the final inequality of $\leq 1$ to hold. This is why the algorithm is only defined for $\varepsilon\in(0,3]$.
\end{proof}

This claim implies that the array size $\ell$ used by $\textsc{BoxSorter}$ satisfies $\ell \geq (1+2^{k-3}\delta)|\mathcal S|$. Since $\textsc{Sorter}_{k-4}$ can, by the inductive hypothesis, solve online sorting with $|\mathcal S|$ numbers on an array of size $(1+2^{k-3}\delta)|\mathcal S|$, this means that $\ell$ is large enough as an array size to insert $|\mathcal S|$ elements using $\textsc{Sorter}_{k-4}$. So it correctly works as our $\textsc{BoxSorter}$ algorithm. 

This concludes the proof that the algorithm $\textsc{Sorter}_k$ can place $n$ numbers on an array of size $(1+2^{k+1}\delta)n$.

\subsubsection{Cost Analysis}
We are left to bound the total cost, for which we use induction. By $\text{cost}^{(k)}(r_1,\ldots,r_n)$, we denote the cost incurred by algorithm $\textsc{Sorter}_k$ when facing the sequence of real numbers $r_1,\ldots,r_n$. We prove that there exists $C>0$ such that $\text{cost}^{(k)}(r_1,\ldots,r_n)\leq \beta n^{1/\omega_k}\delta^{-C(k+1)}$ for any sequence of numbers $r_1,\ldots,r_n$ with $r_i\in [\alpha,\alpha+\beta)$. For $k=1$, the algorithm from \cite{aabk} implies that $\textsc{Sorter}_1$ places an input sequence $r_1,\ldots,r_n\in[\alpha,\alpha+\beta)$ with a cost of $\text{cost}^{(1)}(r_1,\ldots,r_n)=18\beta \sqrt n$, and we can choose $C$ accordingly. 

Now suppose $k>1$, and let $\{B_1,\ldots,B_\ell\}$ be the set of boxes in the array. First we bound the number $\ell$ of boxes, by plugging in the parameters $w=\lfloor (1+2^k\delta)n'\rfloor$, $n'=\left\lfloor\frac{2^{k-1}\delta}{1+2^k\delta}n^{\omega_{k-1}/\omega_k}\right\rfloor$, and $b=\left\lfloor n^{\omega_{k-4}/\omega_k}\right\rfloor$ into the following 
\begin{align*}
\ell & \leq \frac{(1+2^{k+1}\delta)n}{w} = \frac{(1+2^{k+1}\delta)n}{\lfloor(1+2^k\delta)n'\rfloor} \leq 2\frac{(1+2^{k+1}\delta)n}{(1+2^k\delta)n'} \leq \frac{4n}{n'} = \frac{4n}{\left\lfloor\frac{2^{k-1}\delta}{1+2^k\delta}n^{\omega_{k-1}/\omega_k}\right\rfloor} \\ &\leq \frac{8n}{\frac{2^{k-1}\delta}{1+2^k\delta}n^{\omega_{k-1}/\omega_k}}=\frac{8}{\frac{2^{k-1}\delta}{1+2^k\delta}}n^{(\omega_{k}-\omega_{k-1})/\omega_k} = \frac{1+2^k\delta}{2^{k-4}\delta }n^{\omega_{k-4}/\omega_k} \tag{since $\omega_k=\omega_{k-1}+\omega_{k-4}$}\\ 
&\leq  2\frac{1+2^k\delta}{2^{k-4}\delta }\left\lfloor n^{\omega_{k-4}/\omega_k}\right\rfloor=\frac{1+2^k\delta}{2^{k-5}\delta}b.
\end{align*} 
where we repeatedly use the fact that $x\leq 2\lfloor x\rfloor$ for all $x\geq 1$. 

We compute the following two costs separately, which add up to the total cost.
\begin{itemize}
    \item \textbf{Cost inside boxes.} Here, we focus on the cost generated by the elements inside a certain box and sum these up across all boxes. We can think of our algorithm as partitioning the sequence $r_1,r_2,\ldots,r_n$ into subsequences according to the box in which each element is placed. For $i\in[\ell]$, denote the subsequence of length $L_i$ placed in $B_i$ with $y_1^i,\ldots,y_{L_i}^i$ so that we have $\{r_1,\ldots,r_n\}=\bigcup_{i=1}^\ell\{y_1^i,\ldots,y_{L_i}^i\}$. Moreover, $L_i\leq n'$ for each $i\in[\ell]$. Define $\alpha_i'\coloneqq (i-1)\cdot \beta b^{-1}$ and $\beta_i'\coloneqq \beta b^{-1}$, then for each $j\in [L_i]$ it holds $\alpha_i'\leq y_j^i<\alpha_i'+\beta_i'$. By the induction hypothesis, the cost induced by the recursive call of $\textsc{Sorter}_{k-1}$ on box $B_i$ is bounded by 
\begin{align*} 
\text{cost}^{(k-1)}(y_1^i,\ldots,y_{L_i}^i)&\leq \beta_i'\cdot (L_i)^{1/\omega_{k-1}} \delta^{-Ck} \\ 
& \leq \beta b^{-1}\cdot (n')^{1/\omega_{k-1}} \delta^{-Ck}\\ 
& \leq \beta b^{-1}\cdot \left(\frac{2^{k-1}\delta}{1+2^k\delta}n^{\omega_{k-1}/\omega_k}\right)^{1/\omega_{k-1}} \delta^{-Ck} \\ 
& \leq \beta b^{-1}\cdot \left(\frac{2^{k-1}\delta}{1+2^k\delta}\right)^{1/\omega_{k-1}}n^{1/\omega_k}\delta^{-Ck} \\  
&\leq b^{-1}\cdot \beta n^{1/\omega_k} \delta^{-Ck}\tag{since $\frac{2^{k-1}\delta}{1+2^k\delta}<1$}
\end{align*}
Since there are at most $\ell$ boxes in total, we bound the total cost generated inside the boxes $B_i$ as \begin{align*}
\sum_{i=1}^\ell\text{cost}^{(k-1)}(y_1^i,\ldots,y_{L_i}^i) 
&\leq \ell b^{-1} \cdot \beta n^{1/\omega_k} \delta^{-Ck} \\ 
&\leq \frac{1+2^k\delta}{2^{k-5}\delta}\cdot \beta n^{1/\omega_k} \delta^{-Ck} \\ 
&=\left(2^{5-k}\delta^{-1}+8\right)\cdot \beta n^{1/\omega_k}\delta^{-Ck} \\ 
&\leq \frac{1}{2}\beta n^{1/\omega_k}\delta^{-C(k+1)} 
\end{align*}
 where the last inequality holds if we choose $C$ large enough. 

\item \textbf{Cost between boxes.} Now we estimate the cost incurred by $\textsc{BoxSorter}$ which is an instance of $\textsc{Sorter}_{k-4}$. Here, we place $|\mathcal S|$ elements from $[\alpha,\alpha+\beta)$. From \cref{boxsort} and the bound for $\ell$ above, we first note that 
\begin{align*}
    |\mathcal S|\leq \frac{\ell}{1+2^{k-3}\delta}\leq \frac{1}{2^{k-3}\delta}\frac{1+2^k\delta}{2^{k-4}\delta} n^{\omega_{k-4}/\omega_k}= \left(2^{7-2k}\delta^{-2}+2^{7-k}\delta^{-1}\right) n^{\omega_{k-4}/\omega_k}.
\end{align*}

We use the above bound on $|\mathcal S|$ and the inductive hypothesis to bound the total cost in this case as:
\begin{align*}
\text{cost}^{(k-4)}(\mathcal S)&=\beta\cdot |\mathcal S|^{1/\omega_{k-4}}\delta^{-C(k-3)} \\ &\leq \beta \cdot \left(\left(2^{7-2k}\delta^{-2}+2^{7-k}\delta^{-1}\right) n^{\omega_{k-4}/\omega_k}\right)^{1/\omega_{k-4}}\delta^{-C(k-3)} \\ &=\beta \cdot \left(2^{7-2k}\delta^{-2}+2^{7-k}\delta^{-1}\right)^{1/\omega_{k-4}}\cdot n^{1/\omega_k}\delta^{-C(k-3)} \\ &\leq \beta \cdot\sqrt{2^{7-2k}\delta^{-2}+2^{7-k}\delta^{-1}}\cdot n^{1/\omega_k}\delta^{-C(k-3)}\\ &\leq \frac{1}{2}\beta n^{1/\omega_k}\delta^{-C(k+1)}
\end{align*}
where the last inequality holds if we choose $C$ large enough. 
\end{itemize}

Summing up these two costs: \emph{cost inside boxes} and \emph{cost between boxes}, we complete the induction. This concludes the proof that the total cost incurred by $\textsc{Sorter}_k$ for online sorting with $n$ numbers on an array of size $(1+2^{k+1}\delta)n$ is bounded by $\beta n^{1/\omega_k}\delta^{-O(k)}$.

\paragraph{Acknowledgments} We would like to thank Ioana Bercea and L\'{a}szl\'{o} Kozma for introducing us to this problem at the \emph{Algorithms and Data Structures Today} workshop at the National University of Singapore.


\bibliographystyle{alpha}
\bibliography{references}

\end{document}